\documentclass[12pt]{amsart}
\usepackage{amsfonts,latexsym,amsthm,amssymb,graphicx}
\usepackage[all]{xy}
\DeclareFontFamily{OT1}{rsfs}{}
\DeclareFontShape{OT1}{rsfs}{n}{it}{<-> rsfs10}{}
\DeclareMathAlphabet{\mathscr}{OT1}{rsfs}{n}{it}

\CompileMatrices

\newtheorem{theorem}{Theorem}[section]
\newtheorem{lemma}[theorem]{Lemma}
\newtheorem{corol}[theorem]{Corollary}
\newtheorem{prop}[theorem]{Proposition}

{\theoremstyle{definition} }
{\theoremstyle{remark} \newtheorem{remark}[theorem]{Remark}
\newtheorem{example}[theorem]{Example}}

\newcommand{\Qbb}{{\mathbb{Q}}}
\newcommand{\Pbb}{{\mathbb{P}}}
\newcommand{\Zbb}{{\mathbb{Z}}}

\newcommand{\cA}{{\mathscr A}}
\newcommand{\cD}{{\mathscr D}}
\newcommand{\cE}{{\mathscr E}}
\newcommand{\cL}{{\mathscr L}}
\newcommand{\cM}{{\mathscr M}}
\newcommand{\cO}{{\mathscr O}}
\newcommand{\cZ}{{\mathscr Z}}

\newcommand{\one}{1\hskip-3.5pt1}
\newcommand{\csm}{{c_{\text{SM}}}}

\newcommand{\integ}[3][X]{\int_{#3}\one(#2)\,d\mathfrak c_#1}

\newcommand{\hD}{{\widehat D}}
\newcommand{\oD}{{\overline D}}
\newcommand{\uD}{{\underline D}}
\newcommand{\uDelta}{{\underline \Delta}}
\newcommand{\oS}{{\overline S}}
\newcommand{\uS}{{\underline S}}

\DeclareMathOperator{\rk}{rk}
\DeclareMathOperator{\codim}{codim}

\title{
Chern class identities from tadpole matching in  type IIB and F-theory
}
\author{Paolo Aluffi}
\author{Mboyo Esole}
\address{
Mathematics Department, 
Florida State University,
Tallahassee FL 32306, U.S.A.
}
\email{aluffi@math.fsu.edu}
\address{
Afdeling Theoretische Fysica,
Celesijnenlaan 200D,
3001 Heverlee, Belgium
}
\email{Mboyo.Esole@fys.kuleuven.be}

\begin{document}

\begin{abstract}
In light of  Sen's weak coupling limit of F-theory as a type IIB orientifold, 
the compatibility of the tadpole conditions  leads to a non-trivial identity 
relating the Euler characteristics of an elliptically fibered Calabi-Yau
fourfold and of certain related surfaces.

We present the physical argument leading to the identity, and a 
mathematical derivation of a Chern class identity which confirms it, 
after taking into account singularities of the relevant loci. 
This identity of Chern classes holds in arbitrary dimension, and 
for varieties that are not necessarily Calabi-Yau. 

Singularities are essential in both the physics and the
mathematics arguments: the tadpole relation may be interpreted as
an identity involving stringy invariants of a singular hypersurface, and
corrections for the presence of pinch-points. The mathematical
discussion is streamlined by the use of Chern-Schwartz-MacPherson
classes of singular varieties. We also show how the main identity may be
obtained by applying `Verdier specialization' to suitable constructible
functions.
\end{abstract}

\maketitle


\section{Introduction}\label{intro}
Orientifold compactifications of Type IIB string theory on a Calabi-Yau
threefold in the presence of D3 and D7 branes can be geometrically
described by F-theory compactified on a Calabi-Yau fourfold
\cite{Vafa:1996xn, Sen:1997gv}.  Type IIB is defined in a
ten-dimensional Minkowski space while F-theory requires two additional
dimensions, which provide a geometric description of the axion-dilaton
field of type IIB as the complex structure of an elliptic curve (a
two-torus).  Solutions of type IIB at weak coupling usually have a
constant axion-dilaton field. F-theory provides a description of
solutions with a variable axion-dilaton field by allowing the elliptic
curve to be non-trivially fibered over a threefold. This construction
provides a beautiful identification of S-duality in type IIB as the
modular group of the elliptic curve.

Type IIB and F-theory are both severely constrained by `tadpole
conditions' which ensure that the total D-brane charges in a compact
space vanish as required by Gauss's law. Tadpole conditions realize
the physics wisdom according to which, in a compact space, the total
charge should vanish since fluxes cannot escape to infinity.  From a
dynamical perspective, tadpole conditions are consistency requirements
obtained from the local equations of motion and/or the Bianchi
identities by integrating them over appropriate compact spaces.  The
computation of tadpole conditions requires a detailed account of all
contribution to the D-brane charges. This is closely related to anomaly 
cancellations since the presence of Chern-Simons terms in the
D-brane action (needed for the cancellation of chiral and tensor
anomalies) implies that a D-brane usually carries lower brane
charges. In particular, in type IIB, a seven-brane has an induced D3
charge proportional to the Euler characteristic of the complex surface
(a cycle of real dimension four in the Calabi-Yau threefold) on which
it is wrapped.  In F-theory, the induced D3 charge is proportional to
the Euler characteristic of the Calabi-Yau fourfold.  It follows that
in the absence of other sources of D3 charge (like for example non-trivial
fluxes), the consistency of the F-theory/type IIB tadpole relations
leads to relations between the Euler characteristics of
the F-theory fourfolds and the surfaces wrapped by the seven-branes.
 
The link between type IIB orientifolds and F-theory is clearly expressed in the 
case of $\Zbb_2$-orientifold symmetry by Sen's weak coupling limit of F-theory 
\cite{Sen:1997gv}.  When the fourfold is realized as an elliptic fibration over a 
threefold, and Sen's weak coupling limit is used to produce the associated 
Calabi-Yau threefold, the relations can be recovered, as we show in this paper, 
at the price of dealing with singularities of the loci arising in the limit.
In this paper we analyze one representative class of examples of this situation, 
presenting both the physical argument leading to the relation (\S\ref{phy}), 
and a mathematical derivation of an identity of Chern classes which 
implies it (\S\ref{math}). 
In its form arising from physical considerations, the relation has the
following shape. Let $\varphi: Y\to B$ be an $E_8$ elliptic fibration 
over a nonsingular threefold~$B$, and assume that $Y$ is a Calabi-Yau
variety.
Following Sen (\cite{Sen:1997gv}), we can associate with $Y$ a
Calabi-Yau threefold $X$, obtained as the double cover $\rho: X \to B$ 
ramified along a nonsingular surface~$O$; at `weak coupling limit',
the discriminant of $Y\to B$ determines an orientifold supported on $O$,
and a D7-brane supported on a surface $D$ in $X$. 
Comparing the D3 tadpole condition as seen in F-theory and in type IIB  
leads to the (tentative) relation 
\begin{equation*}
\tag{$\dagger$}
2\,\chi(Y) \overset?= \chi(D) + 4\, \chi(O)
\end{equation*}
among Euler characteristics (cf.~\S\ref{phytadc}). However,
the surface $D$ is singular, and singular varieties admit several
  possible natural notions of `Euler characteristic'; it is not {\em a priori\/}
clear which one  should be employed as $\chi(D)$ in 
($\dagger$). By contrast $Y$ and~$O$ are both nonsingular, and $\chi(Y)$, 
$\chi(O)$ must refer to the usual topological Euler characteristic.

Analyzing this situation with mathematical tools, we can prove 
(Theorem~\ref{tadpole}) that in fact the relation ($\dagger$) holds at 
the level of total homology Chern classes, provided that suitable 
correction terms
are factored in to account for the singularities of $D$:
\begin{equation*}
\tag{$\ddagger$}
2\, \varphi_* c(Y) = \pi_* c(\oD) + 4\, c(O) - \rho_* c(S)
\end{equation*}
implying
\begin{equation*}
\tag{$\dagger'$}
2\, \chi(Y) = \chi(\oD) + 4\, \chi(O)-\chi(S)\quad.
\end{equation*}
Here, $\oD\to D$ is a resolution of $D$ (in fact, its normalization), $\pi:
\oD\to B$ is the composition $\oD \to D \to B$, and $S$ stands 
for the pinch locus of $D$. If $\dim B=3$ (the case of 
physical interest), $\chi(S)$ simply counts the number of pinch-points of $D$. 
However (and surprisingly) ($\ddagger$) holds in arbitrary dimension,
and independently of Calabi-Yau hypotheses. Thus, it appears that the
scope of the tadpole conditions is actually substantially more general than the 
context in which they arise.

The term $\pi_* c(\oD)$ could be interpreted as the (push-forward to $B$ 
of the) {\em stringy Chern class\/}
of $D$, and the question remains of whether the corrected class
$\pi_* c(\oD)-\rho_* c(S)$, resp.~its degree $\chi(\oD)-\chi(S)$, are 
mathematically
`natural' notions. We take a stab at this question in \S\ref{spec}, where a
mechanism is proposed which appears to account precisely for the needed
correction term, at least in the class of examples considered in this paper.
We point out that the ingredients used to define the stringy
Chern class (as in \cite{MR2183846}) may in fact be employed to define
other notions of Chern classes and Euler characteristics $\chi^{(m)}$ for 
singular varieties, depending on a parameter~$m$. 
Each value of this  parameter corresponds to a different candidate for
`relative canonical divisor' of the resolution map; for the examples considered 
in this paper, $m=1$ corresponds to the notion leading to stringy invariants,
while $m=2$ corresponds to an alternative notion (leading to `arc' invariants; 
the $\Omega$-flavor considered in \cite{MR2183846}). As we will see in 
\S\ref{spec}, $\chi^{(m)}$ admits a well-defined limit as $m\to \infty$;
and this Euler characteristic $\chi^{(\infty)}$ recovers precisely the relation 
($\dagger$) proposed by the string-theoretic considerations 
(Theorem~\ref{MF}). Therefore, this  appears to be the natural notion in the 
context of this problem.

However, there is room for surprises, and it is not impossible that in more
general situations the singularities modify the tapdole relation in a different 
ways  than we have anticipated here.
Although this will have no effect on the correctness of the mathematical result 
of this paper  (the physics providing just an ansatz from the mathematical 
perspective) it would surely reshape the physics.  
To settle the issue, a complete physical derivation of the tapdole conditions 
taking into account the  singularities would be appropriate.  
Tadpole conditions  in string theory are usually obtained by a  loop calculation or 
by using the inflow mechanism. Both roads have their shortcoming  in presence 
of singularities\footnote{A loop calculation requires a  definition of the field 
theory in presence of  singularities. This is  usually possible when the singularities
are very mild  like for example if they are of the orbifold type.     Performing an 
inflow calculation based on  index theorems is also not free of additional 
assumptions since it will require extending the usual index theorem to  singular
varieties.  Such an extension would depend on the choice of regularization of the
singularities or in other words on the choice of definitions for the  topological 
invariant  of a  singular variety. As we have seen, several non-equivalent choices 
are possible.}.
It is therefore refreshing to know that the point of view presented in this paper 
is  corroborated by  an analysis of different physical aspects of the  system 
(\cite{ACDE}).  In particular it is shown in (\cite{ACDE})  that  the choice of 
the Euler characteristic $\chi^{(\infty)}$ is compatible with a ``deconstruction'' 
description of the  brane configuration   in terms of  a system of D9-anti-D9 
branes with appropriate  world-volume fluxes turn on.\vskip 6pt

This paper is aimed at both physicists and mathematicians; \S\ref{phy}
is written with the former public in mind, and~\S\ref{math} with the latter.
In order to enhance readability, these sections are essentially independent
of each other, and the hurried reader of one sort may ignore the section
more squarely 
written for the other. But the most interesting aspect of our results lies in 
the interplay between the two viewpoints.\vskip 6pt

We include in this introduction a few slightly more technical comments.
The Chern class identity we prove (Theorem~\ref{tadpole}) holds
in arbitrary dimension, for varieties which are not necessarily Calabi-Yau,
and generalizes in the simplest possible way the relation among Euler
characteristics predicted by the tadpole condition presented in 
\S\ref{phytadc}. 
In fact, Theorem~\ref{tadpole} is established by lifting to this level
of generality the elementary {\em Sethi-Vafa-Witten formula\/}
(formula (2.12) in~\cite{MR1421701}) for the 
Euler characteristic of the fibration~$Y$, and comparing it with an 
analogous formula obtained at Sen's weak coupling limit. We view
the Sethi-Vafa-Witten formula as a general statement equating the
Euler characteristic of the fibration with twice the Euler characteristic
of a specific divisor $G$ in the base $B$ of the fibration 
(Proposition~\ref{SVWform}). 
The arguments in \S\ref{math} are streamlined by using the calculus
of {\em constructible functions,\/} which encodes the good properties of
topological Euler characteristic and (by a result of R.~MacPherson) of
Chern classes. The relevant facts are recalled in \S\ref{math}. As a 
concrete example, we offer the following instance of the situation 
considered in this article.

\begin{example}\label{23328}
A degree~24 hypersurface with equation $y^2=z^3 + fz+g$ in weighted 
projective space $\Pbb_{1,1,1,1,8,12}$ (with $y$, resp.~$z$ of degree 
$12$, resp.~$8$, and $f$, $g$ general polynomials in the other variables) 
determines a Calabi-Yau elliptic fibration $Y\to B$, with $B=\Pbb^3$. 
Standard methods (for example, judicious use of the adjunction formula) 
easily yield $\chi(Y)=23328$.
 Sen's weak coupling 
limit leads us to consider surfaces $O$, resp.~$\uD$ in $\Pbb^3$
with equation 
$h=0$, resp.~$\eta^2-12 h \chi=0$, where $h$, $\eta$, $\chi$ 
are general polynomials of degrees~$8$, $16$, $24$ respectively. 
Again, adjunction yields immediately that $\chi(O)=304$; as for $\uD$,
the presence of $8\cdot 16\cdot 24=3072$ nodes at the intersection
$S$ given by $h=\eta=\chi=0$ has to be taken into account, and gives 
$\chi(\uD)=28864-3072=25792$.

The associated Calabi-Yau threefold~$X$ at weak coupling limit is the 
double cover of $\Pbb^3$ branched over $O$; $D$ is the inverse image 
of $\uD$ in $X$. The orientifold and D7 brane are localized on $O$ and 
$D$, respectively.

A simple local analysis shows that $D$ is singular along a curve, with
a set $S$ of $3072$ pinch-points corresponding to the set $\uS$ of 
nodes on $\uD$. It also shows that $D$ may be resolved to a nonsingular 
surface $\oD$ by blowing up the singular curve; the composition 
$\pi: \oD \to \uD$ is found to be $2$-to-$1$ everywhere except over 
$\uS$. In terms of constructible functions, this says that the push-forward 
of the constant function~$\one_{\oD}$~is
$$\pi_*(\one_{\oD})=2\cdot \one_{\uD}-\one_\uS\quad:$$
the function which equals $2$ on $\uD\smallsetminus \uS$ and $1$ on 
$\uS$. It follows then immediately that
$$\chi(\oD)=2\chi(\uD)-\chi(\uS) = 2\cdot 25792-3072 = 48512\quad;$$
what's more, the same relation must hold among the total Chern classes
of these loci (cf.~property (2)~in \S\ref{math}). 

We stress that
more sophisticated intersection-theoretic tools are {\em not\/} needed 
in order to extract this information from the blow-up $\oD \to D$.
These simple considerations suffice to verify the tadpole relation with 
pinch-point correction in this example:
$$2\cdot \chi(Y)=2\cdot 23328=46656=48512+4\cdot (-304)-3072
=\chi(\oD)+4\chi(O)-\chi(S)\quad.$$
\end{example}
Note that ignoring  singularities would have led us to an embarrassing 
tadpole mismatch between type IIB and F-theory at weak 
coupling\footnote{If we consider the case of several D7 branes, one can 
show that there is a unique configuration which satisfies the tadpole 
relation. Namely,   two D7  branes wrapped around  two  smooth surfaces  
given respectively by polynomials of  degree 28 and 4 in the  Calabi-Yau 
three-fold.
However, such a configuration is not generic in type IIB. 
Moreover, in  F-theory, it seems not to be compatible with  Sen's description 
of the weak coupling limit. We will come back to this configuration in 
(\cite{ACDE}), see also section \ref{entract} of the present paper.}.
The proof presented in \S\ref{math} for the general case (at the level of 
Chern classes, in arbitrary dimension and without Calabi-Yau hypotheses) 
is no harder---modulo some intersection theory---than the proof 
sketched above for Example~\ref{23328}. Indeed, the key observation in the proof
of Theorem~\ref{tadpole} is precisely the same formula 
$\pi_*(\one_{\oD})=2\cdot \one_{\uD}-\one_\uS$ used above,
which is just as easy to prove in general as in Example~\ref{23328}. 
In~\S\ref{math} we also offer an alternative, slightly more sophisticated 
viewpoint on such relations of constructible functions (by means of 
{\em Verdier specialization\/}, see Remark~\ref{Verdier}) as a possible 
venue for more general results.     
 Intersection-theoretic invariants of singular varieties, including some
 computations involving blow-ups, also play a role in
 \cite{Andreas:1999ng}, in a comparison between the 
 $\mathrm{E}_8\times \mathrm{E}_8$ heterotic string and F-theory with  G-fluxes.
\vskip 12pt

{\em Acknowledgments.} The first-named author would like to thank the
Max-Planck-Institut f\"ur Mathematik in Bonn, where most of this
project was carried out.  The second-named author would like to thank
A. Collinucci, F. Denef, E. Diaconescu and M. Soroush for valuable
discussions. In particular he would like to thank F. Denef for
introducing him to the D3 brane mismatch in type IIB/F-theory.  He is
also grateful to the Max-Planck-Institut f\"ur Mathematik in Bonn, the
Yang Institute of Theoretical Physics at Stonybrook and the Fifth
Simons Workshop in Mathematics and Physics for hospitality.  This work
was supported in part by NSA grant H98230-07-1-0024, and by the
European Community's Human Potential Programme under contract
MRTN-CT-2004-005104 `Constituents, fundamental forces and symmetries
of the universe'.


\section{Physics}\label{phy}

In this section we present in some detail the physical motivation for
($\dagger$).  Roughly speaking, it results from a direct comparison of
the D3 brane tadpole condition in type IIB and in F-theory in the
situation where no fluxes are turned on.  We will introduce the
necessary notions in a pedagogic way from (\S\ref{TypeII}) to
(\S\ref{tadcons}); readers familiar with Sen's weak coupling limit of
F-theory, tadpole conditions, and dualities, could consider jumping
immediately to (\S\ref{phytadc}), without great harm. 
In (\S\ref{TypeII}), we give some basic notions on D-branes in type II
string theories; in (\S\ref{Dualities}) we review the
  (S-T-U-)dualities among type IIA, type IIB and M-theory. These
dualities are important to understand the dictionary between the
physics and the geometry of F-theory and M-theory; they also provide
an elegant derivation of the F-theory D3 brane tadpole condition.
F-theory is introduced in (\S\ref{Ftheory}) and Sen's weak coupling
limit is reviewed in (\S\ref{E8fibs}).  Tadpoles and anomalies are
discussed in (\S\ref{tadcons}); the case of a $\Zbb_2$ orientifold of
type IIB with O7-planes and D3 and D7 branes as well as the D3 tadpole
condition in F-theory are worked out in detail.
 
In (\S\ref{phytadc}), we derive a general form of the relation
({$\dagger$}) from \S\ref{intro}
with several possible D7 branes and O7-planes. 
Relation ($\dagger$), with only one D7 and one O7 branes, is the
generic situation in Sen's weak coupling limit. However, since
singularities are necessarily present
in Sen's weak coupling limit, ($\dagger$) must be modified to take
them into account.  Assuming that the final answer keeps the same
shape, we give evidence that the modified formula would be of type of
the relation ({$\dagger'$}). This will be confirmed in \S\ref{math}, 
where we analyze Sen's template situation in its natural
mathematical setting, and prove formula ({$\ddagger$}) (which implies
($\dagger'$)) for a larger class of varieties.
For those who are mostly interested in the Calabi-Yau case, we note
that ($\dagger'$) could be proved just by mimicking the treatment of
example \ref{23328}, using no more than the adjunction formula in the
spirit of (\cite{MR1421701}).  
While this computation is straightforward, the material in \S\ref{math}
 provides a deeper understanding of the geometry of the situation.

\subsection{Type II string theories and D-branes}\label{TypeII}
There are five consistent ten-dimen\-sional string theories.  Here we
will mostly be interested in type IIA and IIB string theories.  Each
of these two theories admit 32 supersymmetry generators organized into
two ten-dimensional Majorana-Weyl spinors with opposite chirality in
type IIA and the same chirality in type IIB.  Both theories contain in
their spectrum the following NS-NS ({\em Neveu-Schwarz\/}) fields: a
graviton, an antisymmetric two-form which couples to the fundamental
string, and a scalar field $\phi$ (called the {\em dilaton\/}) which controls the
string coupling $g_s$ in each of these theories, following the relation
$g_s \sim e^{-\phi}$.  They also contain RR ({\em Ramond-Ramond\/})
$(p+1)$-forms $C_{(p+1)}$ with $(p+1)$ odd in type IIA and even in
type IIB.  As a direct generalization of the charged particle in
Maxwell theory, a $(p+1)$-form naturally couples to an object extended
in $p$-spatial dimension. Indeed, as it evolves in spacetime, a
$p$-brane draws a world-volume $W^{(p+1)}$ of spacetime dimension
$(p+1)$ on which the $(p+1)$ potential $C_{(p+1)}$ can be evaluated as
{\small $\int_{W^{(p+1)}} C_{(p+1)}$}.  A $p$-brane charged under a
$(p+1)$-form admits a magnetic dual which is a $(d-p-4)$-brane, where
$d$ is the spacetime dimension of the ambient space in which the brane
lives. They are related by Hodge-conjugation $*F_{(p+2)}=F_{d-p-2}$
acting on their field strengths $F_{p+2}=\mathrm{d}C_{(p+1)}$.

The objects that carry the RR charges are not seen in perturbative
string theory.  It was realized by Polchinski
(\cite{Polchinski:1995mt}) that $p$-branes are actually naturally
present in string theory as loci on which open strings can end. For
that reasons, they are usually called {\em Dirichlet $p$-branes } or
 ({\em D$p$-branes}) since fixing the location of the ends of open
strings is realized by imposing Dirichlet boundary conditions.
 D-branes are {\em half-BPS objects}, which means they preserve only
half of the total amount supersymmetry.  A string with its two ends on
the same D-brane defines a $\mathrm{U}(1)$ gauge field (a 
$\mathrm{U}(1)$ bundle with a connection) on the world-volume of the 
brane. This is the {\em Chan-Paton bundle}.  Note that since it is defined 
from an open string, the Chan-Paton gauge field defined on the 
world-volume of a D-brane is a NS-NS field.
 
 Type IIA admits only D$p$-branes with $p$ even while $p$ is odd in type
IIB.  More precisely, in type IIA we have a D0 brane (also called the
{\em D-particle}) and a D2 brane, their magnetic duals
are respectively the D6 and the D4 brane.
There is also a D8 brane, which does not admits a magnetic dual.
  In type IIB there are a D(-1)
(or {\em D-instanton\/}), D1 (or {\em D-string\/}), D3, D5, D7 and D9
branes. The D7 and the D5 are the magnetic duals of the D-instanton
and the D-string. The D9 does not admit a magnetic dual in ten
dimensions. The fundamental string which is present both in type IIA
and type IIB is usually called the {\em F-string,\/} and couples
electrically to the NS-NS two-form. The magnetic dual of the F-string
is a 5-brane (the {\em NS5 brane}) which is present both in type IIA
and type IIB.

\subsection{Dualities and M-theory}\label{Dualities}
The five ten-dimensional string theories are related by a web of
dualities which also include eleven-dimensional supergravity. The
latter is the supersymmetric field theory of gravity with the highest
possible spacetime dimension for the usual Minkowski signature.  We
shall review quickly some of these dualities in order to understand
the origin of F-theory.

{\em T-duality} identifies the physics of type IIA compactified on a
circle of radius $r_A$ and type IIB compactified on a circle of radius
$r_B$ provided that these two radii are inverse of each other when
measured in string units. A $p$-brane wrapped around the T-duality
circle is T-dual to a $(p-1)$-brane not intersecting the T-duality
circle and vice versa.

{\em S-duality} is a symmetry which relates the weak coupling regime
of one theory to the strong coupling regime of another one.  It opens
a perturbative window in the strong coupling regime of a theory.  Type
IIB is its own S-dual. More precisely, the S-duality group in type IIB is
a $\mathrm{SL}(2,\Zbb)$ group of type IIB \cite{Hull:1994ys}.  The
S-dual theory of type IIA is an eleven-dimensional theory. The
 radius of the additional eleventh dimension grows as the coupling
constant of type IIA increases.  This eleven-dimensional theory is  
  called M-theory and at low energy it is described by eleven-dimensional 
supergravity. The latter admits a three-form potential
$A_{(3)}$ which can naturally couple to a membrane.  This is the {\em
M2-brane\/} in M-theory and its magnetic dual is 5-brane called the {\em
M5-brane.\/}

{\em U-duality} relates type IIB and M-theory by using a combination
of type IIA/IIB T-duality and type IIA/M-theory S-duality. More
precisely, it provides a duality between type IIB string theory
compactified on a circle $S^1$ and M-theory compactified on a torus
$T^2=S^1\times S^1$, where the first circle is type IIA T-duality
circle (of inverse radius than the type IIB radius) and the second one
is the M-theory circle that controls the string coupling of type IIA.
{\em Type IIB on a circle is dual to M-theory compactified on a torus
whose area is shrinking to zero. }  The modular group of this torus
precisely corresponds to the $\mathrm{SL}(2,\Zbb)$ S-duality group of
type IIB theory. This geometrization of S-duality is the main interest
of F-theory.

\subsection{F-theory}\label{Ftheory}
We consider type IIB compactified on a Calabi-Yau threefold with D3
and D7-branes.  All the branes are taken to be spacetime filling: they
fill all the four-dimensional spacetime and are wrapped around a cycle
of the Calabi-Yau threefold.  A D3-brane will be point-like in the
extra six dimensions and a 7-brane will wrap around a complex surface
 of the compact space while filling the four-dimensional spacetime.
In type IIB we have two type of strings: the D-string has a RR charge
while the F-string has a NS-NS charge.  More generally, a
$(p,q)$-string is the bound state of $p$ F-strings and $q$ D-strings
\cite{Witten:1995im} with $p,q$ relatively prime
\cite{Schwarz:1995dk}; a $(p,q)$-brane is a brane on which
$(p,q)$-strings can end \cite{Douglas:1996du,Gaberdiel:1997ud}.  The
usual D-brane is a $(1,0)$-brane.

The two real scalar-fields of type IIB are organized into a complex
axion-dilaton field
\begin{equation*}
\tau=C_{(0)}+\mathrm{i} e^{-\phi},
\end{equation*}
where the axion is the RR-scalar $C_{(0)}$, and $\phi$ is the dilaton
coming from the NS-NS sector; we recall that $e^{\phi}$ is the string 
coupling constant. The $\mathrm{SL}(2,\Zbb)$ symmetry of type IIB  
acts on the axion-dilaton field as a modular transformation:
\begin{equation*}
\tau \rightarrow \frac{a\tau +b}{c\tau +d}   \, ,\quad 
\begin{pmatrix}
a & b \\
c & d 
\end{pmatrix}
\in \mathrm{SL}(2,\Zbb).
\end{equation*}
The $(p,q)$-strings are the $\mathrm{SL}(2,\Zbb)$ images of the
fundamental string \cite{Schwarz:1995dk,Douglas:1996du}.
The existence of the 
$\mathrm{SL}(2,\Zbb)$ symmetry of type IIB forces us to
contemplate the occurrence of $(p,q)$-branes for any relatively prime
$(p,q)$. However, this would require going above the usual weak
coupling limit of type IIB, since the axion-dilaton field also changes
under $\mathrm{SL}(2,\Zbb)$ without preserving the scale of the
string coupling constant: strong and weak coupling can be mapped into
each other.

F-theory is a description of type IIB string theory in the presence of
$(p,q)$ 7-branes.  These branes are non-perturbative objects which
require a non-constant axion-dilation field.  Since the axion-dilaton
field $\tau$ is subject to modular transformations, Vafa
\cite{Vafa:1996xn} has proposed to describe it as the complex
structure of an elliptic curve. 
{\em It is conjectured that F-theory on an
elliptically fibered Calabi-Yau fourfold with a section and a base
$B$ is equivalent to type IIB on the base $B$ with $(p,q)$ 7-branes at
the singular loci of the elliptic fibration.\/} The 7-branes are wrapped
around divisors of the base $B$.  

The modular group in F-theory is also the same as the modular group of
the torus used to define U-duality between type IIB and M-theory. It
follows that F-theory on an elliptically fibered Calabi-Yau is dual to
M-theory on the same manifold in the limit where the fiber has a
vanishing area.  The three-form $A_{(3)}$ of M-theory reduces to the
NS-NS and RR two-forms of type IIA under S-duality. These two-forms
 under T-duality give the NS-NS two-form of type IIB and the RR
one-form of type IIB.
  We can then conclude that U-duality between type
IIB and M-theory implies that an M2 brane wrapping the two torus
defined by the T-duality circle of type IIA and the S-duality circle
of M-theory will give rise to $(p,q)$-strings in F-theory.

\subsection{The weak coupling limit of F-theory  with $E_8$ 
fibrations}\label{E8fibs}
An $\mathrm{E}_8$ elliptic fibration $\varphi: Y\rightarrow B$ has a
Weierstrass normal equation
\begin{equation*}
x y^2-(z^3+f z x^2+g x^3)=0,
\end{equation*}
written in a $\mathbb{P}^2$ bundle $\phi:\mathbb{P}(\cE)
\rightarrow B$. Here, $f$ and $g$ are respectively sections of
powers $\cL^4$, $\cL^6$ of a line bundle $\cL$ on the base $B$ 
(cf.~\S\ref{intromath}). The variety $Y$ is a Calabi-Yau ($K_Y=0$) if
$c_1(\cL)=c_1(TB)$.  
For every point of the base, the Weierstrass equation of the elliptic
fibration defines an elliptic curve with (Klein's) modular function
\begin{equation*}
j(q)=4 \cdot\frac{ (24 f)^3}{\Delta}.\label{modular}
\end{equation*}
The function $j$ is the generator of modular functions of weight 
one, and $\Delta$ is the discriminant of the elliptic curve:
\begin{equation*}
\Delta=4 f^3+27  g^2.
\end{equation*}
F-theory on the elliptically fibered Calabi-Yau fourfold $Y$ is 
(conjecturally) equivalent 
to type IIB on the base $B$ with $q=e^{2\pi i \tau}$. Since
$\Im(\tau)=e^{-\phi}=\frac{1}{g_s}$ is the inverse of the string
coupling constant, it follows that weak coupling ($g_s\ll 1$)
corresponds to small $q$ and therefore to large $j$ since in the 
limit of small $q$ we have $j\approx q^{-1}$.  The 7-branes are
located at points of the base manifold $B$ where the elliptic fiber is
singular (\cite{Sen:1997gv}); this is where $j(\tau)$ becomes
infinite. These surfaces correspond to the vanishing locus of the
discriminant $\Delta$. {\em A priori,\/} the discriminant locus may 
have several components~$\Delta_i$ which correspond to several 
$D7$-brane worldvolumes.
Perturbative string theory on IIB background has two $\Zbb_2$
symmetries: the world sheet parity inversion $\Omega$ and the
left-moving fermion number $(-1)^{F_L}$.
Given a Calabi-Yau threefold $X$ admitting an involution $\sigma$, 
we can mod out the spectrum of type IIB by the orientifold projection
\begin{equation*}
\Omega \cdot (-1)^{F_L}\cdot  \sigma^*, 
\end{equation*}
where $\sigma$ acts on the type IIB field via its pulback $\sigma^*$.
In order to have only D3 and D7 branes, the involution $\sigma$ is
required to be holomorphic with the additional property $\sigma^*
\Omega^{3,0}=-\Omega^{3,0}$, where $\Omega^{3,0}$ is the holomorphic
three-form of the Calabi-Yau three-form $X$.  Under the action of
$\sigma$, the fixed locus is made of complex surfaces and/or isolated
points. They correspond respectively to  {\em orientifold 7-planes\/} 
(O7 planes)  and O3 planes.

Following Sen (\cite{Sen:1997gv}) we set
\begin{align*}
f & = -3h^2 + c\eta,\nonumber\\
g & = -2 h^3 + c h \eta +c^2 \chi,
\end{align*}
where $c$ is a constant and $h,\eta,\chi$ are respectively general 
sections of
line bundles $\cL^2$, $\cL^4$, $\cL^6$ where $\cL$ is the anticanonical
bundle of $B$, as above.
We recall that large $j$ corresponds to large $\Im(\tau)$ and
therefore to weak coupling. The limit $c\rightarrow 0$ is called the
{\em weak coupling limit} since then $j(\tau)$ is large at every point
of the base except in sectors where $|h|^2\sim |c|$.  Since in the
weak coupling limit we have
\begin{align*}
\Delta \approx -9 c^2  h^2(\eta^2+12h\chi),\label{DeltaF}
\end{align*}
the zeroes of the dominant term of $\Delta$ are supported on $h=0$.  
Sen shows (\cite{Sen:1997gv}) that $h=0$ determines the locations of the
$O7$-planes, while the surface $\uD$ with equation $\eta^2+12 h\chi=0$ 
determines the locations of the $D7$-branes.

One can define a type IIB orientifold equivalent to the weak coupling
limit of F-theory starting with a Calabi-Yau threefold $X$ which is the
double cover of the base $B$ branched along $h=0$ (\cite{Sen:1997gv}):
\begin{equation*}
x_0^2=h,\label{SenCY}
\end{equation*}
where $x_0$ is a section of the line bundle $\cL$.  The
$\Zbb_2$-isometry that is gauged to describe the orientifold is
$x_0\rightarrow -x_0.$ The fixed points under this symmetry correspond
to the hypersurface $h=0$. The variety $X$ has vanishing first Chern 
class, and is therefore a Calabi-Yau manifold.

\subsection{Tadpole conditions}\label{tadcons}
It is natural to consider the surface $D\subset X$ obtained as 
inverse image of $\uD\subset B$. In order to analyze D7-branes
localized on $D$, we study the following general set-up.

We consider a D-brane wrapped around a cycle $D$ of a Calabi-Yau
manifold $X$.  Open strings with both end points on the D-brane define
the Chan-Paton bundle $E\rightarrow D$.
The charge of the D-brane will depend on the embedding 
$f :D\hookrightarrow X$ and the topology of the Chan-Paton bundle $E$. 
This charge can be computed by an anomaly inflow argument
(\cite{Minasian:1997mm}).  An anomaly is a violation of a symmetry of
the Lagrangian by quantum effects; the anomaly of a gauge symmetry
indicates an inconsistency of the theory and must vanish.  The anomaly
inflow mechanism \cite{Callan:1984sa,Cheung:1997az}
consists of introducing an anomalous term in the
Lagrangian to cancel the anomalous contribution of another term. The
two terms will be anomalous when considered separately, but together
they give an anomaly free theory.

Since string theory is anomaly free, self-consistency requires that
for each possible anomaly there is a contribution in the Lagrangian
that maintains the theory anomaly free.  In other words, identifying a
possible anomaly is an opportunity to discover a new sector of the
Lagrangian of the theory. 
The new terms coming from anomaly inflow have been recovered by direct 
string theory calculations 
\cite{Craps:1998fn,Craps:1998tw,Morales:1998ux,Scrucca:1999uz}.
In the case of D-brane configurations, an anomaly can be generated by
massless chiral fermions or self-dual tensor fields located on the
intersection of two branes. The cancellation of this chiral anomaly
requires the presence of an anomalous term, usually called a
 Chern-Simons term or a Wess-Zumino term.  The chiral anomaly is first
 computed using an index theorem. The Chern-Simons term is then deduced
by a descent procedure. All these steps are purely algebraic and are
by now well understood \cite{AlvarezGaume:1984dr,Scrucca:1999uz}. 
Since the Chern-Simons term is linear in the
RR potential, it gives a charge to the RR fields.

The Lagrangian for a RR field $C_{(p+1)}$ is of the type
\begin{equation*}
{\mathcal L}=-\frac{1}{4}F\wedge *F + J \cdot C_{(p+1)} ,
\end{equation*}
where $C_{(p+1)}$ is a $(p+1)$-form and $F=\mathrm{d} C_{(p+1)}$ is
its field strength and $*F$ its Hodge dual; $J$ is called the 
{\em current.\/}
The equation of motion of $C$ is
\begin{equation*}
\mathrm{d}(* F)=J.
\end{equation*}
Using Gauss's law, the charge is the integral of the current. If we
are in a compact space, the equations of motion tell us that the total
charge should vanish:
\begin{equation*}
Q_{\text{Total}}= \int J=\int \mathrm{d}(* F)=0.
\end{equation*}
This necessary condition is usually called a {\em tadpole condition.}   
See \cite{Rabadan:2002wy} for a pedagogic introduction.

\subsubsection{Tadpoles in Type IIB}\label{TadpoleB}
The Chern-Simons term for a D-brane wrapping a cycle $D$   
(admitting a $\mathrm{Spin^c}$ structure\footnote{See section 4.3 of 
(\cite{Witten:1998cd}) or \S 3 of (\cite{ABS}).}) embedded in a 
Calabi-Yau threefold $X$ ($f:D\hookrightarrow X$) with a Chan-Paton 
bundle $E$  is given by\footnote{We consider the case where the 
NS-NS two-form vanishes and its field-strength has no discrete torsion.} 
(\cite{Minasian:1997mm}):
\begin{equation*}
\int_X Q_D (f_*  E)\wedge C=\int_D   \mathrm{ch}(E')   \wedge
\sqrt{\frac{\hat{\mathrm{A}}(TD)}{\hat{\mathrm{A}} (ND)}}  \wedge f^* C ,
\quad E'=E\otimes K_D^{-\frac{1}{2} },
\end{equation*}
where $C=C_0+C_2+C_4+C_6+C_8\in H^{\rm{even}}(X)$ is the total RR
potential and $f^* C $ its pullback to $D$; $\hat{A}$ is the total
A-roof genus (\cite{Harvey:2005it}); $TD$ is the tangent bundle to $D$
and $ND$ is its normal bundle; $\mathrm{ch}(E')$ is the total Chern character of 
twisted sheaf 
{\small $E'=E\otimes
K_D^{-\frac{1}{2}}$} where $K_D$ is the canonical bundle of $D$.  The
appearance of the term {\small $K_D^{-\frac{1}{2}}$} is related to the
Freed-Witten anomaly \cite{Freed:1999vc,Katz:2002gh}. 
When the cycle wrapped by the D-brane is not a  $\mathrm{Spin}$ but a  
$\mathrm {Spin}^c$ manifold,  spinors are not section of the spin bundle 
$\mathrm{Spin}(D)$: such a bundle will suffer from a $\Zbb_2$ ambiguity. 
This ambiguity is cancelled by taking the tensor product with  {\small 
$K_D^{-\frac{1}{2}}$} since the latter admits the same ambiguity ($K_D$ is 
always odd  for a $\mathrm{Spin^c}$ manifold). This amounts to replacing
$E$ with   the twisted bundle {\small $E'=E\otimes K_D^{-\frac{1}{2}}$}. 
Spinors are then sections not of $\mathrm{Spin}(D)$ but of the well-defined 
bundle  $\mathrm{Spin}(D)\otimes E'$. 
In presence of $\Zbb_2$-torsion in $H^2(D,\Zbb)$, the canonical 
bundle {\small $K_D$} will admit more than one  square root, and therefore 
there would be many possible $\mathrm{Spin}^c$ structure on $D$
(\cite{Freed:1999vc}).  
Moreover, given a line bundle $\cM$, one can also replace {\small $K_D$} by 
{\small $K_D\otimes \cM^2$} so that $E'$ becomes  $E'\otimes \cM$.  
This reflects the freedom to choose  different Chan-Paton bundles on a 
D-branes  (\cite{Witten:1998cd}).
In particular,  when the manifold is $\mathrm{Spin}$, one can choose  
{\small $\cM^2=K_D$}  so that the charge formula depends only on  $E$. 
In this paper, when computing charges,  we will always refer to the canonical 
$\mathrm{Spin}^c$ lift  {\small $E'=E\otimes K_D^{-\frac{1}{2}}$}, even when 
$D$ is $\mathrm{Spin}$-manifold.\\ 
  
There is also a  Chern-Simons term for an orientifold plane
wrapped around a cycle $O$ embedded in $X$ as $i:O\hookrightarrow X$
(see for example \cite{Scrucca:1999uz,Blumenhagen:2006ci}):
\begin{equation*}
\int_X Q_{O}\wedge C=-2^{p-4}\int_O  
\sqrt{ \frac{ \hat{\ \mathrm{L}}(TO/4)}{ \hat{\  \mathrm{L}} (NO /4)}} 
\wedge  i^* C ,
\end{equation*}
where $\hat{\ \mathrm{L}}(S)$ is the Hirzebruch
polynomial\footnote{For any $d$, $\hat{\ \mathrm{L}}(dE)$ is defined
as $\sum_k d^{k} \hat{\ \mathrm L}_{k}(E)$ where $\hat{\ \mathrm
L}_k(E)$ is the term of degree $k$ in the expansion $\hat{\ \mathrm
L}(E)=\sum_j \hat{\ \mathrm L}_{j}(E)$.} of $S$ and $i^* C$ is the
pullback of total RR potential $C$ to $O$.

The Chern-Simons term of a D$p$-brane or an O$p$-plane involves the
total RR potential. It follows that a given $p$-brane has not only a
$p$-brane charge but induces as well lower brane charges.  The charge
induced by the Chern-Simons term defines an element of $H^*(X,\Zbb)$
called the {\em Mukai vector\/}. In the previous formulae $Q_D$ and
$Q_O$ are respectively the Mukai vector of a D-brane and and O-plane
wrapped respectively around a complex surface $D$ and $O$.  In type
IIB with spacetime filling branes, the component of the Mukai vector
of degree $n$ represent the induced D$(9-n)$ brane charge.

For a D7 with a trivial Chan-Paton bundle $E'$ and an O7 brane wrapped
respectively around a complex surface $D$ and $O$ of a Calabi-Yau
threefold $X$, we get the following Mukai vectors:
\begin{align*}
 Q_D &=[D]+
\frac{\chi(D)}{24}  \omega,
\\
Q_O &=-8[O]+\frac{\chi(O)}{6}\omega.
\end{align*}
where $\omega$ is the unit volume element of $X$. We have used $\int_S
\mathrm{c_2}(S)=\chi(S)$ and
\begin{equation*}
\hat{\mathrm{A}}(S)=1-\frac{1}{24}(c_1^2-2 c_2),\quad 
\hat{\ \mathrm{L}}(S)=1+\frac{1}{3}(c_1^2-2c_2),\quad c_i=c_i(S).
\end{equation*}

For a single D3 brane, we have 
\begin{equation*}
Q_{D3}=-\omega,
\end{equation*}
where $\omega$ is the volume density of the Calabi-Yau three-fold $X$
(with $\int_X \omega=1$) dual to a point; the conventional minus sign
ensures that the D3 tadpole can be solved by introducing D3 branes of
positive charge \cite{ACDE}.
In a $\Zbb_2$ orientifold configuration without fluxes with $N_{D3}$
D3-branes, D7-branes wrapping divisors $D_i$ with trivial Chan-Paton
  vector bundle, and O7-planes wrapping around divisors $O_j$, the {\em
tadpole condition\/} (cancellation of charges) reads:
\begin{align*}
\text{D7 tadpole }  & : \sum_i [D_i]-8\sum_j [O_j]=0,\\
\text{D3 tadpole }  & :N_{D3}=\frac{1}{2} \left(\sum_i  
\frac{\chi(D_i)}{24}+4 \sum_j \frac{\chi(O_j)}{ 24}\right),
\end{align*}
where the indices $i$ and $j$ label respectively the D7-branes and the 
O7-planes.
The factor of $\frac{1}{2}$ in the D3 tadpole takes into account the
double counting of D3 charge in the cover space of a $\Zbb_2$
Calabi-Yau orientifold. 
Note also that the conventional minus sign in the charge of a single
D3 brane ($Q_3=-\omega$) ensures that the induced D3 charge coming
from the curvature of the 7-branes is cancelled by $N_{D3}$
D3-branes and not by $N_{D3}$-anti-D3 branes.

\subsubsection{Tadpole in F-theory}\label{TadpoleF}
F-theory compactified on an elliptically fibered Calabi-Yau four-fold
$Y$ admits a D3-tadpole condition which is obtained by a sequence of
string dualities \cite{MR1421701}.  D3-branes are the only branes in
type IIB invariant under $\mathrm{SL}(2,\Zbb)$. Therefore, in contrast
to $(p,q)$ 7-branes, D3 branes in F theory are essentially the same
D3 branes seen in type IIB. In a sense, D3 branes play the same
role as M2-branes in M-theory and fundamental strings in type IIA
string theory.  When an M2-brane is wrapped around the eleventh
dimension used to relate M-theory and type IIA, it reduces to the
F-string of type IIA, while an M2 brane that does not intersect the
eleventh dimension of M-theory will give a D2-brane in type IIA.  More
precisely, the three-form of M-theory reduces to the NS-NS two-form
that couples to the F-string of type IIA while the transverse part of
 $A_{(3)}$ reduces to the RR-three form $C_{(3)}$ that
couples to the D2-brane.  Moreover, under T-duality, a D3-brane
wrapped around the T-duality circle will give a D2-brane in type IIB,
under S-duality this D2 brane corresponds to a M2 brane transverse to
the S-duality circle of M-theory. The determination of the F-theory
tadpole can be simply deduced by reading the following sequence of
dualities:
\begin{center}
\begin{tabular}{ccccc}
  & & & & \\
IIA  & $  \overset{S-duality}{\longrightarrow}$& M theory &   $\overset{U-duality}{\longrightarrow}$   & Type IIB\\
F-string & & M2-brane & & D3 brane\\
$\int_{M_{2}\times Y} B_{(2)} \wedge Y_8$
& & $\int_{M_{3}\times Y} A_{(3)} \wedge Y_8$
& & $\int_{M_{4}\times Y} C_{(4)} \wedge Y_8$\\
& &  & & \\
\end{tabular}
\end{center}
where $Y_8$ is a characteristic class for the four-fold $Y$ such that
$\int_Y Y_8=\frac{\chi(Y)}{24}$ when $Y$ is a Calabi-Yau\footnote{More
precisely, we have $Y_8=-\frac{1}{192}(c_1^4-4c_1^2 c_2 +8 c_1 c_3-8
c_4)$, where here the Chern classes $c_i$ are those of $Y$.  When $Y$
is a Calabi-Yau, $c_1=0$ and $Y_8=\frac{c_4}{24}$.}
(\cite{Becker:1996gj}) and $M_d$ represents the $d$-dimensional
spacetime.  Compactification of type IIA string theory on a Calabi-Yau
four-fold $Y$ to two dimensions leads to a tadpole term for the NS-NS
two-form $B_{(2)}$ that couples to the fundamental string
\cite{Vafa:1995fj}. This tadpole is proportional to the Euler
characteristic of $Y$.  Since the corresponding type IIA interaction
$\int B_{(2)}\wedge Y_8$ does not depend on the dilaton, it can be
lifted to M-theory using S-duality, but the NS-NS two-form should be
replaced by the three-form $A_{(3)}$ which couples to the M2 brane
 (\cite{Becker:1996gj,Duff:1995wd}). This new interaction $\int
A_{(3)}\wedge Y_8$ can be seen as a quantum correction to the
classical Chern-Simons term $\int A_{(3)}\wedge \mathrm{d}
A_{(3)}\wedge \mathrm{d} A_{(3)}$ of eleven-dimensional supergravity.
If we assume that there are no fluxes, the vanishing of the tadpole
requires the presence of $N_{M2}$ M2 branes, so that
 $N_{M2}=\frac{\chi(Y)}{24}$ (see chapter 10 of \cite{Becker:2007zj}).
Finally, using U-duality between M-theory and F-theory, there is a
similar tadpole for F-theory compactified on the Calabi-Yau four-fold
$Y$, but this time for the four-form $C_{(4)}$ which couples to the D3
brane.  The tadpole in type IIA is cancelled by the presence of NS-NS
charge, in M theory it is cancelled by the presence of M2-branes
charge while in F-theory compactified on an elliptically fibered
Calabi-Yau four-fold $Y$, the tadpole is cancelled by D3 brane
charge. If the latter is solely coming from $N_{D3}$ D3 branes, it
gives (\cite{MR1421701}):
\begin{equation*}
N_{D3}=\frac{\chi(Y)}{24}.
\end{equation*}

\subsection{Matching F-theory and type IIB tadpole conditions}\label{phytadc}
Consistency between type IIB  and the F-theory D3 tadpole implies that 
\begin{equation*}
2\chi(Y)=\sum_i \chi( D_i)+4 \sum_j \chi(O_j),
\end{equation*}
by simply equating the expressions obtained for the number of D3
branes ($N_{D3}$) required in these two theories to satisfy the
D3-tadpole condition. It is interesting to note that the two sides of 
this equality involve objects defined in different
regimes. The elliptically fibered Calabi-Yau four-fold $Y$ is
introduced to describe regimes in which the string coupling can be
strong in presence of possible $(p,q)$ 7-branes; on the other hand,
the l.h.s.~of the equality involves solely O-planes and $(1,0)$ D7 
branes which are only well-defined at weak coupling. This is well
illustrated for example in (\cite{Sen:1997gv}), where an O7-plane
is shown to correspond in strong coupling to a system of $(p,q)$ 
7-branes that coincide when the coupling becomes weak enough. The usual
monodromy of the axion-dilaton field around such an O7-plane is
reproduced as the total monodromy around the corresponding system of
$(p,q)$ 7-branes at strong coupling. It is therefore natural to
consider the previous relation in a weak coupling limit of F-theory.

In Sen's weak coupling limit for $\mathrm{E_8}$ elliptic fibrations, and
for general choices of $h$, $\eta$, $\chi$ (and hence of $f$, $g$), we
have a unique orientifold plane $O$ and a unique D7 brane $D$ in
type IIB. Arguing as above, consistency between type IIB and F-theory 
D3 tadpole would give the equality presented in the introduction:
\begin{equation*}
\tag{$\dagger$}
2\chi(Y)\overset?=\chi(D)+4 \chi(O)\quad.
\end{equation*}

However, equation ($\dagger$) should be parsed carefully. The general arguments
in \S\ref{tadcons} assume implicitly that all cycles under exam are
{\em nonsingular\/} (for example, the expression of the Chern-Simons
term assumes that the tangent bundle $TD$ exists), while this is
{\em not\/} the case for the surface $D$ supporting the D7 brane
in Sen's weak coupling limit. In the base $B$, the D7-brane is on the 
surface $\uD$ defined by the equation
$\eta^2+12h\chi=0$. In the the Calabi-Yau threefold $X$, the inverse
image $D$ of $\uD$ is defined by the equation
\begin{equation*}
\eta^2+12 x_0^2 \chi=0.
\end{equation*}
This surface is {\em singular\/} along the double curve $\eta=x_0=0$,
and has pinch points   (cf.~\cite{MR507725}, p.~617)
  at $\eta=x_0=\chi=0$.

There are in general several `natural' definitions of Euler
characteristic (or Chern class) of a singular variety, all giving the
 ordinary topological notions when applied to a nonsingular variety 
 (see for example the appendix of  \cite{ACDE}).
The Euler characteristics of the Calabi-Yau fourfold $Y$ and of the
orientifold $O$ are unambiguously defined, since these varieties are
assumed to be nonsingular; 
but it is not clear how the term
$\chi(D)$ should be interpreted in equation~($\dagger$). Turning
things around, we could consider equation ($\dagger$) as giving a
`prediction' for $\chi(D)$.  It is then natural to ask if this
prediction matches a natural definition of Euler characteristic for a
more general singular variety.

We will formulate some more concrete speculations along these lines 
in~\S\ref{spec}.

Another viewpoint on this situation is that, for more conventional
Euler characteristics, the relation ($\dagger$) should only be
satisfied modulo a contribution from the singularities of $D$, which
would vanish in the smooth case. The task amounts then to evaluating
this correction term precisely.

This is accomplished in \S\ref{math}, with the result stated in the
introduction: adopting the Euler characteristic of the {\em
normalization\/} $\oD$ of $D$ as the natural notion of Euler
characteristic for the singular surface $D$, we will find that the
correction term needed in order to recover the tadpole relation
($\dagger$) is {\em evaluated by the number of pinch-points of $D$.\/}
Note that $\oD$ is nonsingular, and that it can be identified as the
blow-up of $D$ along its singular locus.

Example~\ref{23328} shows this phenomenon at work in a concrete
instance, and the reader should have no difficulties adapting the
proof given there to analyze the general case considered in
\S\ref{E8fibs}, reaching the same conclusion.
  The physics underlying this particular example is analyzed in detail
in (\cite{ACDE}).


\section{Entr'act}\label{entract}

Now that we have set the stage, we can address one doubt that may be
lingering in the mind of the reader: is it truly necessary to invoke
the presence of singularities in order to verify the tadpole
condition? Might there not exist simpler configurations, consisting of 
D-branes supported on nonsingular surfaces, and simply
 satisfying the Euler characteristic constraints imposed by the tadpole
condition?

In general (but with one notable exception, see below)
this appears not to be the case: the relations are
 known not to hold when applied to examples where all loci
are assumed to be smooth (\cite{ACDE}).

In F-theory, the seven-branes only wrap surfaces over which the
elliptic fibration is singular.  This restricts seriously the allowed
configurations.  
For example, if we restrict ourself to  $\Zbb_2$-orientifolds, in type IIB, 
the typical D7 configuration is   
composed of an O7-plane and a D7-brane; they wrap two complex
surfaces that intersect along a curve.  The D7 tadpole condition
defines a linear relation among the homology cycles of these two
surfaces. 

In the case at hand, the O7 is supported on a smooth surface of class 
$c_1(\cL)$ in the Calabi-Yau threefold $X$; the D7 tadpole condition 
forces the D7 to be supported on a surface or  an union of surfaces of 
total class $c_1(\cL^8)$.

Now, if we assume that the D7 is supported on a single, smooth 
surface, then adjunction shows immediately that the tadpole matching 
condition of IIB and F-theory will simply {\em not\/} be satisfied.
From a type IIB  perspective,  without any input from Sen's sharp 
description of the orientifold limit of F-theory, this would have been the 
typical  configuration and would have not satisfied the tadpole matching 
condition of IIB and F-theory. However, this mismatch can be attributed 
to a naive identification of the typical configuration. 

Assuming that the D7 is wrapped on a union of general smooth
surfaces of varying classes equal to multiples of $c_1(\cL)$, one
can verify (again using adjunction) that the matching is obtained
for precisely one configuration: the O7-plane and two D7 branes wrapped 
 around surfaces of classes $c_1(\cL)$, $c_1(\cL^7)$, respectively.
 However, this configuration does not seem to be compatible with Sen's 
description of the orientifold limit of F-theory. 
 Thus, singularities in the support of the D7 brane do appear to be a necessary 
feature of the situation, at least from the point of view of Sen's limit.
As we have illustrated in Example~\ref{23328}, and as we are
going to verify in general in \S\ref{math}, it is possible to satisfy
 the tadpole conditions within Sen's description, if we take seriously
the presence of singularities.

The identity we will obtain by doing so will in fact realize the tadpole
conditions at the level of Chern classes, and in arbitrary dimension.
The configuration of two smooth surfaces mentioned above does 
 not appear to generalize in the same fashion. This is further evidence
that the
configuration cannot be produced by a geometric construction
analogous to Sen's description.


\section{Mathematics}\label{math}

\subsection{}\label{intromath}
We consider the following situation, extending the set-up of \S\ref{E8fibs}.
Let $B$ be a nonsingular compact complex algebraic variety of any dimension, 
and let $\varphi: Y \to B$ be an elliptic fibration, realized by a Weierstrass normal 
equation
\begin{equation*}
\tag{*}
y^2 x - (z^3 +f\,z x^2+g\, x^3) = 0
\end{equation*}
in a $\Pbb^2$-bundle\footnote{We use the projective bundle of 
{\em lines\/} in $\cE$.} $\phi:\Pbb(\cE) \to B$. Here (as in \S\ref{E8fibs}) $f$, 
resp.~$g$ are sections of powers $\cL^4$, resp.~$\cL^6$ of a line bundle 
$\cL$ on $B$. We can take $\cE=\cO\oplus \cL^3\oplus \cL^2$; the 
left-hand-side of (*) realizes $Y$ as the zero-scheme of a section of 
$\cO_{\Pbb(\cE)}(3)\otimes \phi^*\cL^6$ in~$\Pbb(\cE)$.

We assume that the base loci of the linear systems $|\cL^4|$, $|\cL^6|$
are disjoint, and that $f$, $g$ are general. We let $\Delta\subset B$ denote 
the discriminant hypersurface, given by
$$4\,f^3+27\,g^2 = 0\quad;$$
$\Delta$ is the zero-locus of a section of~$\cL^{12}$.
The following is observed in \cite{MR1243538},~1.5 (cf.~\cite{MR690264}, 
Proposition~2.1; and \cite{MR977771} for Weierstrass models):

\begin{lemma}\label{marku}
With notation and assumptions as above:
\begin{itemize}
\item $Y$ is nonsingular and $\varphi$ is flat;
\item for $p\not\in \Delta$, the fiber $\varphi^{-1}(p)$ is a smooth elliptic curve;
\item for $p\in \Delta$, $f(p)\ne 0$, the fiber $\varphi^{-1}(p)$ is a nodal cubic;
\item for $p\in \Delta$, $f(p)=g(p)=0$, the fiber $\varphi^{-1}(p)$ is a cuspidal
cubic.
\end{itemize}
\end{lemma}

We will denote by $F$, resp.~$G$ the hypersurfaces determined by $f$, $g$, 
and we will assume that $F$ and $G$ are nonsingular and intersect transversally.
We also assume that $\Delta$ is nonsingular away from the codimension~2
locus
$$C:\quad f=g=0\quad,$$
which is a nonsingular variety by the transversality hypothesis.
All these assumptions are satisfied (by Bertini) if e.g., $\cL$ is very ample.
\begin{center}
\includegraphics[scale=.4]{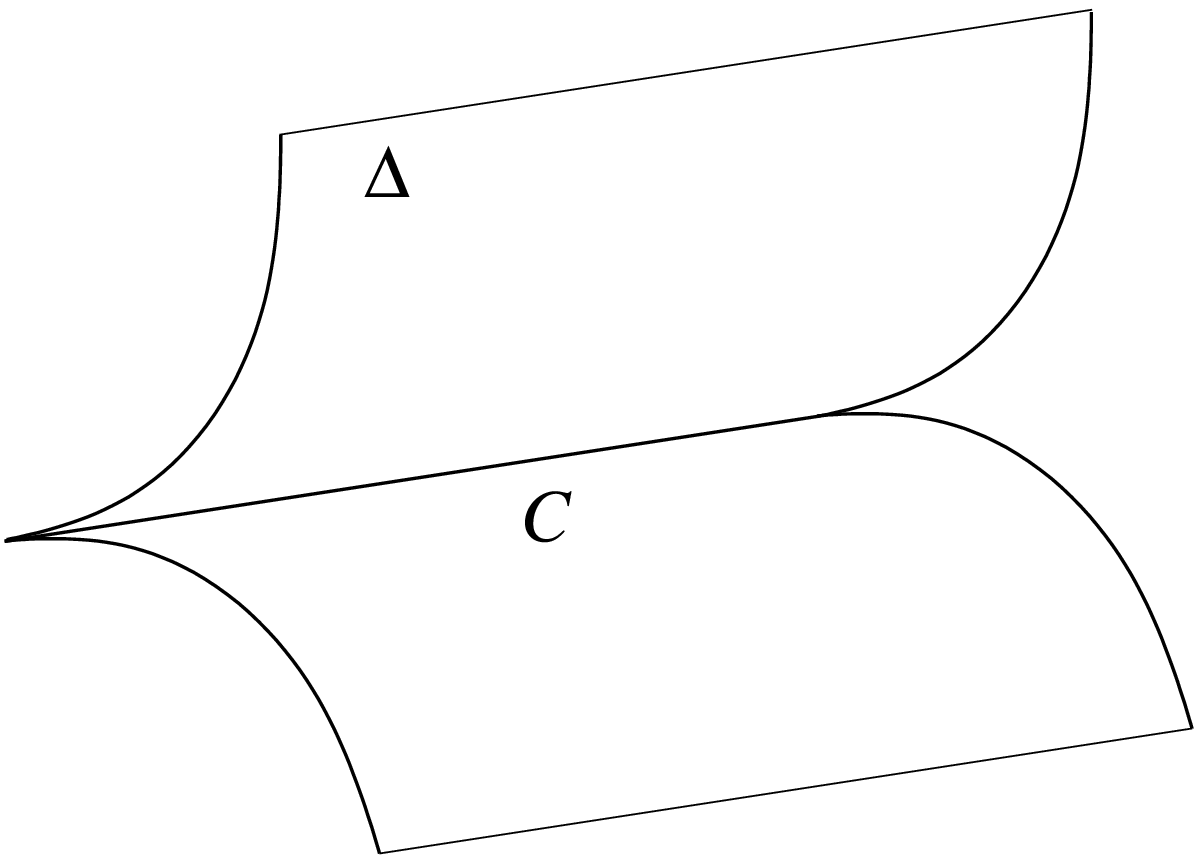}
\end{center}

As noted in \S\ref{E8fibs}, $Y$ is a Calabi-Yau variety if and only if $\cL$ is 
the anticanonical
bundle of~$B$ (\cite{MR1243538}, 1.5~(i)). However, this hypothesis will not 
be needed for the main results in this section.

\subsection{Sethi-Vafa-Witten formula}\label{SVWform}
Our first task is to extend the Sethi-Vafa-Witten formula (cf.~\S\ref{intro})
to the situation described above. We interpret this formula as a comparison
between the total Chern class of the elliptic fibration~$Y$ and the 
total Chern class of the hypersurface~$G$.

\begin{prop}\label{SVW}
Let $\varphi: Y \to B$ be an elliptic fibration, as above. Then
$$\varphi_*\, c(Y)=2\,\iota_* c(G)\quad,$$
where $\iota: G \to B$ is the embedding of $G$.
\end{prop}

Here and in the following, we denote by $c(X)=c(TX)\cap [X]$ the
total `homology' Chern class of the tangent bundle of $X$, {\em if $X$ is a 
nonsingular variety.\/}

For the proof of Proposition~\ref{SVW}, we rely on the calculus of
constructible functions and Chern-Schwartz-MacPherson (CSM)
classes on (possibly) singular varieties. 
We recall that a {\em constructible function\/} on a (complex,
compact, possibly singular) variety~$X$
is a $\Zbb$-linear combination of characteristic functions of closed
subvarieties of~$X$. Constructible functions on $X$ form an abelian
group $F(X)$, which is covariantly functorial: if $\psi: X \to Y$ is a proper 
 morphism of varieties, then $\psi$ induces a homomorphism $\psi_*:
F(X) \to F(Y)$ uniquely defined by the requirement that if $Z\subset X$
is a closed subvariety, then
$$\psi_*(\one_Z)(p)=\chi (\psi^{-1}(p)\cap Z)\quad,$$
where $\one_Z$ is the characteristic function of $Z$, and $\chi$ denotes 
topological Euler characteristic. To each constructible function 
$\alpha\in F(X)$ we can associate a {\em CSM class\/} in the Chow group 
of $X$:
$$\csm(\alpha)\in A_*(X)\quad,$$
such that
\begin{enumerate}
\item if $X$ is nonsingular, then $\csm(\one_X)$ equals the total Chern
class of $X$:
$$\csm(\one_X)=c(X)\quad;$$
\item the assignment of CSM classes is functorial, in the sense that if
$\psi:X \to Y$ is a proper morphism then $\forall \alpha\in F(X)$
$$\psi_*(\csm(\alpha))=\csm(\psi_*(\alpha))\quad.$$
\end{enumerate}
In view of (1), one defines the {\em Chern-Schwartz-MacPherson class\/}
of a (possibly singular) variety~$X$ as
$$\csm(X):=\csm(\one_X)\in A_*(X)\quad:$$
thus, $\csm(X)=c(X)$ if $X$ is nonsingular, but $\csm(X)$ is defined for
arbitrary varieties~$X$. It easily follows from (2) that the degree of the
zero-dimensional component of $X$ is the topological Euler characteristic
$\chi(X)$ of $X$.

In these definitions, the reader may replace the Chow group with 
ordinary integral homology. This topological setting was the context of the 
original result of MacPherson (\cite{MR50:13587}) and of earlier, independent
work of Marie-H\'el\`ene Schwartz (\cite{MR35:3707, MR32:1727});
MacPherson's and Schwartz's very different definitions are known to lead
to the same class (\cite{MR83h:32011, aluffi-2006}). For a rapid review of the 
definition of $\csm$ see \cite{MR85k:14004}, \S19.1.7. An alternative 
 is given in \cite{MR2209219}, Definition~3.2.

Although the statement of Proposition~\ref{SVW} only involves 
{\em nonsingular\/} varieties, its proof is considerably streamlined by
using the notions we just recalled, which were introduced 
for the study of {\em singular\/} spaces.

\begin{proof}[Proof of Proposition~\ref{SVW}]
By Lemma~\ref{marku} and the definition of push-forward of constructible
functions recalled above, we have
$$\varphi_*(\one_Y)=\one_\Delta+\one_C\quad;$$
by the properties of CSM classes recalled above,
$$\varphi_*c(Y)=\iota_{\Delta*}\csm(\Delta)+
\iota_{C*} c(C)\quad,$$
where $\iota_\Delta$, $\iota_C$ denote the corresponding embeddings
(note that $C$ is nonsingular).

The class $\csm(\Delta)$ may be computed by using Theorem~I.4 from 
\cite{MR2001i:14009}:
\begin{equation*}
\tag{**}
\csm(\Delta)=c(TB)\cap \left(\frac{[\Delta]}{1+\Delta}+\frac 1{1+\Delta} 
\left( s(\Delta_s,B)^\vee \otimes \cO(\Delta)\right)\right)\quad,
\end{equation*}
where $\Delta_s$ denotes the {\em singularity subscheme\/} of $\Delta$ 
(defined locally by the ideal of partial derivatives of an equation for $\Delta$), 
$s(\Delta_s,B)$ is its Segre class in $B$ (cf.~\cite{MR85k:14004}, Chapter~4), 
and $\cO(\Delta)$ is the line bundle of which  $\Delta$ is a section.
By the assumptions detailed at the beginning of the section, $\Delta_s$ is
supported on $C$.
The differential of the (local) equation for $\Delta$ is
$12 f^2\, df + 54 g\, dg$;
the differentials $df$, $dg$ are linearly independent in a neighborhood 
of~$C$, and it follows that $\Delta_s$ has ideal
$$(f^2,g)\quad,$$
that is, it is the complete intersection of $G$ and the `double' $2F$ of $F$.
Since the complete intersection of $F$ and $G$ is $C$, we conclude
$$s(\Delta_s,B)=\frac{2[C]}{(1+2F)(1+G)}\quad.$$  
 Applying (**) we get\footnote{\label{calculus}Here we use the simple calculus of the
operations $\otimes$, ${}^\vee$, see \S2 in \cite{MR1316973}. The
expression
$\frac{2[C]}{(1+2F)(1+G)}$
is viewed as $c(\cA)^{-1}\cap a$, where $\cA$ is a vector bundle
with roots $2F$, $G$, and $a$ is the homology class $2[C]$. Now, for
all vector bundles $\cA$: 
$$(c(\cA)^{-1}\cap a)^\vee=c(\cA^\vee)^{-1}\cap a^\vee\quad,$$
and for all line bundles $\cM$
$$(c(\cA)^{-1}\cap a)\otimes \cM=c(\cM)^{\rk \cA} c(\cA\otimes \cM)^{-1}
\cap (a\otimes \cM)\quad.$$
The classes $a^\vee$ and  $a\otimes \cM$ are both linear in $a$. For a pure-dimensional $a$, they 
equal $(-1)^{\codim a}a$, $c(\cM)^{-\codim a}\cap a$, respectively.}:
\begin{align*}
\csm(\Delta) &=c(TB)\cap \left(\frac{[\Delta]}{1+\Delta}+\frac 1{1+\Delta}
\frac{2[C]}{(1-2F)(1-G)}\otimes \cO(\Delta)\right)\\
& =c(TB)\cap \left(\frac{[\Delta]}{1+\Delta}+\frac 1{1+\Delta} 
\frac{2[C]}{(1+\Delta-2F)(1+\Delta-G)}\right)
\end{align*}
Since $\Delta=3F=2G$ as divisor classes, and $[C]=F\cdot [G]$, this shows
$$\csm(\Delta) =c(TB)\cap \left(\frac{[\Delta]}{1+\Delta}+\frac 1{1+\Delta} 
\frac{2[C]}{(1+F)(1+G)}\right)= c(TB)\cap \frac{(2+F)\cdot [G]}{(1+F)(1+G)}\quad.$$
Using this, and again the fact that $C$ is the complete intersection of $F$ and $G$,
\begin{align*}
\varphi_*c(Y) &=c(TB)\cap \left(
\frac{(2+F)\cdot [G]}{(1+F)(1+G)} +\frac{[F] [G]}{(1+F)(1+G)}\right)\\
&=c(TB)\cap \frac{2\,[G]}{1+G}=2\,\iota_*c(G)\quad,
\end{align*}
giving the statement.
\end{proof}

Since the Euler characteristic agrees with the degree of the top Chern
class, we obtain the following statement (in terms of the class of the line 
bundle $\cL$ introduced at the beginning of the section):

\begin{corol}\label{SVWcor}
Let $Y$ be an elliptic fibration on a nonsingular compact variety $B$,
as above. Let $\ell=c_1 (\cL)$, and $c_i=c_i(TB)$, and let $b=\dim B$. 
Then 
$$\chi(Y)=12\,\ell \left(c_{b-1}-6\,\ell\, c_{b-2}+6^2\, \ell^2\, c_{b-3}
+\cdots + (-6)^{b-1}\,\ell^{b-1}\right)\quad.$$
\end{corol}

\begin{proof}
Since $G$ has class $6\ell$,
$$\iota_*c(G) = c(TB)\,\frac{6\ell}{1+6\ell}\cap[B]\quad.$$
Applying Proposition~\ref{SVW}, and reading off the term of 
dimension~$0$, yields the statement.
\end{proof}

Proposition~\ref{SVW} and Corollary~\ref{SVWcor}~hold regardless of any 
Calabi-Yau hypothesis. 
As observed above, $Y$ is a Calabi-Yau variety if and only if $\ell=c_1$; 
in this case, Corollary~\ref{SVWcor} reproduces Proposition~2 in \S8 of 
\cite{Klemm:1996ts}. Some values for $\chi(Y)$ are given in the table below. 
{\small
\begin{table}[ht]
\renewcommand\arraystretch{1.5}
\noindent\[
\begin{array}{|c||c|c|}
\hline
\dim B & \chi(Y) & \text{for $Y$ a Calabi-Yau ($\ell=c_1$):} \\
\hline
\hline
1 & 12\,\ell & 12\,c_1 \\
\hline
2 & 12\,\ell(c_1-6\, \ell) & 12\,c_1(-5\, c_1)  \\
\hline
3 & 12\, \ell(c_2-6\,\ell c_1+6^2 \ell^2) & 12\, c_1(30\, c_1^2+c_2)  \\
\hline
4 & 12\, \ell(c_3-6\, \ell c_2+6^2\,\ell^2c_1-6^3 \ell^3) & 
12\, c_1(-180\, c_1^3-6\,c_1 c_2+c_3)  \\
\hline
5 & 12\, \ell(c_4-6\,\ell c_3+6^2\,\ell^2 c_2-6^3\, \ell^3 c_1+6^4 \ell^4)
& 12\, c_1(1080\, c_1^4+36\,c_1^2 c_2-6\, c_1 c_3+c_4) \\
\hline
\end{array}
\]
Euler characteristic of $E_8$ elliptic fibrations
\end{table}}
The third line in the table (that is, Corollary~\ref{SVWcor} for $\dim B=3$ and 
$\ell=c_1$) reproduces formula~(2.12) in~\cite{MR1421701}.
In this sense, Proposition~\ref{SVW} should be viewed as a generalization
of the Sethi-Vafa-Witten formula---holding in arbitrary dimension, for
$E_8$~elliptic fibrations that are not necessarily Calabi-Yau varieties,
and at the level of total Chern classes.

\subsection{Tadpole relation}
Our next goal is to analyze the tadpole relation of \S\ref{phytadc} in the context
of the situation presented in \S\ref{intromath}: that is, to obtain a precise
relation for the Euler characteristics, holding in arbitrary dimension, and 
bypassing the Calabi-Yau hypothesis. We first obtain a Chern class relation 
involving the discriminant determined in \S\ref{intromath} and the discriminant 
at weak coupling limit; then we interpret the result in terms akin to those
presented in \S\ref{phytadc}.

Recall that the weak coupling limit is obtained by viewing the defining 
equation~(*)
$$y^2 x - (z^3 +f\,z x^2+g\, x^3) = 0$$
as a perturbation of the degenerate fibration
$$y^2 x - (z^3 +(-3h^2)\,z x^2+(-2h^3)\, x^3) = 0\quad,$$
where $h$ is a general section of $\cL^2$. More precisely, we let
$$\left\{\aligned
f & =-3h^2 + c\,\eta \\
g & = -2h^3 + c\,h\eta +c^2\, \chi
\endaligned\right.$$
where $c$ is a scalar, and $\eta$,~resp.~$\chi$ are general sections of
$\cL^4$, resp.~$\cL^6$. For general~$c$, we are in the situation presented
in~\S\ref{intromath}; the weak coupling limit is obtained by letting
$c\to 0$. The resulting family of discriminants has flat limit
$$h^2 (\eta^2+12h\chi)$$
for $c=0$.
The corresponding hypersurface $\uDelta$ is the union of a (double) nonsingular 
component~$O$ with equation $h=0$, and the singular hypersurface~$\uD$
given by $\eta^2+12 h \chi=0$. Under our standing generality assumptions,
the only singularities of $\uD$ are along the (transversal) intersection
$h=\eta=\chi=0$, a nonsingular subvariety~$\uS$ of codimension~$3$ in $B$. 
In fact, as $d(\eta^2+12 h \chi)=2\eta \,d\eta+12h \,d\chi+12 \chi \,dh$
and $d\eta,d\chi,dh$ are independent along $S$, it follows that the
singularity subscheme $\uD_s$ of $\uD$ coincides with~$\uS$.

We now consider the problem of expressing $\varphi_*(c(Y))$ in terms of 
this limiting discriminant, analogously to Proposition~\ref{SVW}.
The following should be viewed as a `limiting Sethi-Vafa-Witten formula':

\begin{lemma}\label{SVWlimit}
With notation as above,
$$\varphi_*c(Y)=\csm(\one_\uD+2\,\one_O-\one_\uS)\quad.$$
\end{lemma}

\begin{proof}
Since $O$ and $\uS$ are nonsingular, 
$$\csm(O)=c(TB)\cap \frac{[O]}{1+O}\quad,\quad
\csm(\uS)=c(TB)\cap \frac{[\uS]}{(1+O)(1+F)(1+G)}$$
by the normalization property (1) of CSM classes, and noting that 
$\eta$, resp.~$\chi$ are sections of $\cL^4\cong\cO(F)$, resp.~$\cL^6
\cong \cO(G)$.

The class $\csm(\uD)$ is evaluated by again applying~Theorem~I.4 from 
\cite{MR2001i:14009}:
$$\csm(\uD)=c(TB)\cap \left(\frac{[\uD]}{1+\uD}+\frac 1{1+\uD} 
\left( s(\uD_s,B)^\vee \otimes \cO(\uD)\right)\right)\quad.$$
We have already observed that this scheme is in fact $\uS$, hence
$$s(\uD_s,B)=s(\uS,B)=\frac{[\uS]}{(1+O)(1+F)(1+G)}\quad.$$
Applying the `calculus' recalled in footnote~\ref{calculus}:
$$\csm(\uD)=c(TB)\cap \left(\frac{[\uD]}{1+\uD}-\frac 1{1+\uD} 
\frac{[\uS]}{(1+O)(1+F)(1+G)}
\right)\quad.$$
Using that $[\uD]=2[F]=[O]+[G]$, $[G]=[O]+[F]$,
and $[\uS]=[O][F][G]$, it follows that
$$\csm(\uD)+2\csm(O)-\csm(\uS)=c(TB)\cap \frac{2[G]}{1+G} = 
2\, \iota_* c(G)\quad.$$
The statement follows then immediately from Proposition~\ref{SVW}.
\end{proof}

\begin{remark}[Verdier specialization]\label{Verdier}
According to Lemma~\ref{SVWlimit}, the two constructible functions
$$\one_\Delta+\one_C
\qquad\text{and}\qquad
\one_\uD+2\one_O-\one_\uS$$
have the same CSM class in $A_*(B)$. It is natural to ask for `systematic'
ways to produce such identities. One possibility consists of studying
the {\em Verdier specialization\/} $\sigma_X$ of the constant function~$1$ 
from the total space of a smoothing family for a hypersurface~$X$. For an 
efficient summary of this notion, introduced in~\cite{MR629126}, we 
recommend \S5 of \cite{MR2002g:14005}; the function $\sigma_X$ is 
essentially defined by taking Euler characteristics of nearby fibers. In the 
case at hand, the reader can verify that
$$\sigma_{\Delta}=\one_\Delta-2\,\one_C +2\,\one_{C'}$$
and
$$\sigma_{\uDelta} = \one_{\uD}+2\one_O-6\one_Q+3\one_\uS-\one_{O'}+
5\one_{Q'}-3\one_{\uS'}\quad,$$
where $Q=O\cap \uD$, and primed letters 
denote the intersection of the corresponding locus with a general
element of the linear system of the hypersurface 
(use~\cite{MR2002g:14005}, Proposition 5.1). 
Now, it is easy to
see that the CSM class of the specialization function $\sigma_X$
only depends on the divisor class of $X$ (in fact, it reproduces the
Chern class of the virtual tangent bundle of~$X$);
 hence
$$\csm(\sigma_{\Delta})=\csm(\sigma_{\uDelta})\quad.$$
Further, one can verify directly that
$$\csm(6\one_Q-3\one_C+\one_{O'}-4\one_\uS-5\one_{Q'}+2\one_{C'}
+3\one_{\uS'})=0\quad:$$
this is straightforward since every locus appearing on the l.h.s.~is  
a nonsingular complete intersection. Combining these identities yields precisely that
$$\csm(\one_\Delta+\one_C)=\csm(\one_\uD+2\one_O-\one_\uS)\quad.$$
This gives an alternative argument for Lemma~\ref{SVWlimit}, bypassing
the use of~\cite{MR2001i:14009} and shedding some light on the
reason why such an identity should hold in the first place.
\end{remark}

Next, we consider (as in \S\ref{E8fibs})
the double cover $\rho: X \to B$ ramified along the smooth hypersurface~$O$.
Note that $X$ has vanishing canonical class
when $[O]=2\,c_1(TB)\cap [B]$; it follows that $X$ is a Calabi-Yau 
if $Y$ is a Calabi-Yau. 
However, again we point out that this hypothesis is not necessary
for the considerations in this section.

We denote by $D$ the inverse image of $\uD$ in $X$. The local 
analysis summarized in \S\ref{phytadc}
goes through unchanged in the situation considered in this section.
Explicitly: $X$ may be realized (as in \S\ref{phytadc}) by putting
$h=x_0^2$; we may use $\eta, \chi, x_0, x_1,\dots, x_r$ as local
coordinates in $X$, and $D$ is then defined (locally) by $\eta^2
+12 x_0^2\chi=0$. This hypersurface is singular along $\eta=x_0=0$,
corresponding to the inverse image of $Q=O\cap \uD$, and `pinched' 
along $\eta=\chi=x_0=0$, that is, the inverse image $S=\rho^{-1}(\uS)$.
Note that $\rho$ restricts to an isomorphism $S \to \uS$. We consider
the resolution $\oD$ of $D$ obtained by blowing up $\eta=x_0=0$:
locally, the blow-up of $X$ is covered by two charts, and we may choose
local coordinates in one of these charts as follows:
$$\tilde\eta, \chi, \tilde x_0, x_1,\dots,x_r$$
so that the blow-up map is given by
$$(\tilde\eta, \chi, \tilde x_0, x_1,\dots,x_r) \mapsto 
(\tilde \eta\tilde x_0, \chi, \tilde x_0,x_1,\dots,x_r)\quad.$$
In these coordinates, the inverse image of $D$ is given by the equation
$\tilde x_0^2(\tilde \eta^2 +12 \chi)=0$;
therefore the blow-up $\oD$ of $D$ has equation
$$\tilde \eta^2 +12 \chi=0\quad,$$
and in particular it is nonsingular in this chart. The situation on the
other local chart of the blow-up can be analyzed similarly, with the
same conclusion: blowing up the singular locus of $D$ resolves
its singularities\footnote{If $D$ is a surface, as in \S\ref{phytadc},
this is of course the standard resolution of the Whitney umbrella
obtained by blowing up along the singular curve.}.

Note that the map $\oD\to D\to \uD$ is generically 2-to-1, and
1-to-1 precisely over~$\uS$. We denote by 
$$\pi: \oD \to B$$
the composition $\oD\twoheadrightarrow D\twoheadrightarrow
\uD\hookrightarrow B$.

\begin{theorem}\label{tadpole}
With notation as above,
$$2\, \varphi_* c(Y)=\pi_* c(\oD)+4\, c(O)-\rho_* c(S)$$
in $A_*(B)$.
\end{theorem}

\begin{proof}
The map $\pi$ is $2$-to-$1$ onto the complement of $\uS$ in $\uD$, and
$1$-to-$1$ over the singular locus $\uS$ of $\uD$. Therefore, by
definition of push-forward of constructible functions,
$$\pi_*(\one_{\oD})=2\,\one_\uD - \one_\uS\quad.$$
It follows that 
$$\pi_*(\one_{\oD})+4\,\one_O-\rho_*(\one_S) 
= 2\,\one_\uD+4\,\one_O-2\,\one_\uS\quad,$$
and hence, by Lemma~\ref{SVWlimit},
$$\csm(\pi_*(\one_{\oD})+4\,\one_O-\rho_*\one_S )
=2\,\csm(\one_\uD+2\,\one_O-\one_\uS)=2\,\varphi_*c(Y)\quad.$$
The formula given in the statement follows from this by the properties
(1), (2) of CSM classes recalled in \S\ref{SVWform}.
\end{proof}

Considering only the term of dimension~$0$ in Theorem~\ref{tadpole}
gives

\begin{corol}\label{tadcor}
$$2\,\chi(Y)=\chi(\oD)+4\,\chi(O)-\chi(S)\quad.$$
\end{corol}

If $\dim B=3$, so that $S$ consists of a discrete set of points, then
$\chi(S)$ simply equals the number of pinch-points of the surface~$D$.
This relation is the corrected version of the tadpole relation ($\dagger$)
promised in \S\ref{phytadc}. It is generalized to arbitrary dimension, and 
emancipated from the Calabi-Yau hypothesis.


\section{Speculation}\label{spec}

The Euler characteristic $\chi(\oD)$ appearing in Corollary~\ref{tadcor} 
could be interpreted as the {\em stringy Euler characteristic\/} of $D$, 
although $D$ is not normal: the blow-up $\oD\to D$ is `crepant' in 
the sense that its differential is regular in codimension~$1$. 
By the same token, the push forward of $c(\oD)$ should
be interpreted as the {\em stringy Chern class\/} of the singular 
hypersurface $D\subset B$, cf.~\cite{MR2183846} and 
\cite{MR2304329}.

In fact, the machinery of~\cite{MR2183846} produces other `natural'
notions of Chern class (and, in particular, of Euler characteristic) for 
singular varieties, depending on how the relative canonical divisor of a 
resolution is handled. For a review of these notions, the reader
is addressed to \cite{MR2280127}. Briefly: if $\nu: \overline Z\to Z$ is a 
resolution of a singular variety~$Z$, the `celestial integral'
$$\integ[{{\overline Z}}]{K_\nu}{\overline\cZ}$$
determines a class in $(A_*Z)_\Qbb$, which may be taken as defining a
`total Chern class' for~$Z$. Here, $K_\nu$ denotes the chosen notion
of relative canonical divisor of $\nu$, which in turns depends on how
$\omega_Z$ is defined:
\begin{itemize}
\item Taking $\wedge^{\dim Z} \Omega^1_Z$ leads to the
{\em arc Chern class\footnote{This is called `$\Omega$ flavor' in
\cite{MR2183846} and \cite{MR2280127}, as opposed to the (stringy)
`$\omega$ flavor'.}} of $Z$, $c_{\text{arc}}(Z)$;
\item Taking the double-dual of $\wedge^{\dim Z} \Omega^1_Z$ leads to the
{\em stringy Chern class\/} $c_{\text{str}}(Z)$. This notion was
independently defined and studied in \cite{MR2304329}.
\end{itemize}
The degree of the Chern class recovers the corresponding
notion of Euler characteristic. If $Z$ is nonsingular to begin with, all
these notions coincide and simply reproduce the usual Chern class
and (topological) Euler characteristic of $Z$.

For the situation considered in \S\ref{math}, we have the resolution
$$\nu:\oD \to D$$
obtained by blowing up a codimension~1 locus in $D$.
A computation in local coordinates shows that
$\Omega^{\dim D}_{\oD|D}$ is supported on the inverse image 
$\oS=\nu^{-1}(S)$ of the pinch locus of $D$; this is a codimension~$2$
locus, and it follows that the relative canonical divisor in the `stringy'
sense vanishes; therefore, $c_{\text{str}}(D)$ is evaluated by
(the `identity manifestation' of)
$$\integ[{{\oD}}]{0}{\overline \cD}\quad.$$

The resolution~$\oD$ is not adequate to compute $c_{\text{arc}}(D)$,
since $\Omega^{\dim D}_{\oD|D}$ is not invertible (that is,
$\oD$ does not `resolve' the data in the sense of \cite{MR2183846},
\S3.3). In order to obtain a resolution satisfying this condition, we 
blow up $\oD$ further along $\oS$:
let $\hD$ be this blow-up, and denote by $E$ the exceptional divisor;
and let $\hat\nu: \hD \to D$ be the composition of the two blow-ups. 
Then the stringy relative canonical divisor is $E$, while 
$\Omega^{\dim D}_{\hD|D}\cong \cO(2E)$. 

This prompts
us to propose the following (speculative) definition: for every 
nonnegative integer $m$, we can let $c^{(m)}(D)$ be the class in 
$(A_*D)_\Qbb$ corresponding to the celestial integral
$$\integ[{{\hD}}]{mE}{\hat \cD}\quad,$$
so that $c^{(1)}(D)= c_{\text{str}}(D)$ and $c^{(2)}(D)= c_{\text{arc}}(D)$.
As we will see in a moment, the class $c^{(m)}(D)$ has a well-defined `limit as 
$m\to \infty$', which we will denote $c^{(\infty)}(D)$; the corresponding
degrees will be denoted $\chi^{(m)}(D)$, $\chi^{(\infty)}(D)$.

\begin{theorem}\label{MF}
With notation as in above and as in \S3:
$$2\,\varphi_* c(Y)=c^{(\infty)}(D)+4\, c(O)\quad.$$
\end{theorem}

This statement is proved by computing $c^{(m)}$ explicitly; the following
lemma does this, and gives a meaning to the limit of the class $c^{(m)}$
as $m\to \infty$.

\begin{lemma}\label{specula}
With notation as above,
$$c^{(m)}(D)=\pi_* c(\oD) - c(S) + \frac 2{1+m} c(S)\quad.$$
\end{lemma}

\begin{proof}
By definition of the integral (see \cite{MR2280127}, \S7) we get{\small
\begin{align*}
c^{(m)}(D) &:= \hat\nu_*\left(
\frac{c(T\hD)}{(1+E)}
\left( 1+\frac{E}{1+m}\right)
\cap [\hD]\right)
=\hat\nu_*\left(c(T\hD)\cap [\hD]-\frac{m}{1+m_E}c(TE)\cap [E]\right) \\
&=\hat\nu_*\csm \left(\one_{\hD}-\frac{m}{1+m} \one_E\right)\quad.
\end{align*}}
Since $\oS$ is nonsingular and of codimension~2 in $\oD$, the exceptional
divisor $E$ is a $\Pbb^1$-bundle over $\oS$. As $\oS$ maps isomorphically
to $S$, we have
$$\hat\nu_* (\one_E) = 2\cdot \one_S\quad.$$
By the same token,
$$\hat\nu_* (\one_{\hD})= \pi_*(\one_{\oD})+\one_S\quad.$$
Thus
$$\hat\nu_* \left(\one_{\hD}-\frac{m}{1+m} \one_E\right)
=\pi_*(\one_{\oD}) -\one_S + \frac 2{1+m} \one_S\quad,$$
and the statement follows by applying $\csm$ and using its properties
(1), (2) listed in~\S\ref{math}.
\end{proof}

The theorem follows immediately from Lemma~\ref{specula}
and Theorem~\ref{tadpole}.

Theorem~\ref{MF} recovers the tadpole relation of \S\ref{phytadc} 
`without correction terms':
$$2\,\chi(Y)=\chi^{(\infty)}(D)+4\,\chi(O)\quad,$$
in full alignment with the string theory prediction obtained in \S\ref{phytadc}.
However, it is of course unclear at this point whether this is due to a lucky
accident, or whether the formalism leading to the definition of $c^{(\infty)}
(D)$ can really account for the relevant information at a good level
of generality.



\end{document}